\documentclass[american,aps,pra,reprint, superscriptaddress, showpacs]{revtex4-2}
\usepackage[unicode=true,pdfusetitle, bookmarks=true,bookmarksnumbered=false,bookmarksopen=false, breaklinks=false,pdfborder={0 0 0},backref=false,colorlinks=false]{hyperref}
\hypersetup{colorlinks,linkcolor=myurlcolor,citecolor=myurlcolor,urlcolor=myurlcolor}
\usepackage{graphics,epstopdf,graphicx, amsthm, amsmath, amssymb, times, braket, color, bm, physics,comment,bbm}

\definecolor{myurlcolor}{rgb}{0,0,0.7}

\newtheorem{theorem}{Theorem}

\newtheorem{lemma}{Lemma}

\newcommand{\SPAN}{\mathrm{span}}
\newcommand{\cptp}{\mathcal{E}}
\newcommand{\iden}{\mathbbm{1}}
\newcommand{\diag}{\mathrm{diag}}

\newcommand{\cccone}{\mathcal C_{\mathrm{CCTIO}}}
\newcommand{\tiocone}{\mathcal C_{\mathrm{TIO}}}
\newcommand{\tio}{\mathrm{TIO}}

\begin{document}

\title{Amplifying asymmetry with correlating catalysts}

\author{Feng Ding}

\author{Xueyuan Hu}
\email{xyhu@sdu.edu.cn}
\affiliation{School of Information Science and Engineering, Shandong University, Qingdao 266237, China}

\author{Heng Fan}
\affiliation{Institute of Physics, Chinese Academy of Sciences, Beijing 100190, China}
\affiliation{CAS Center for Excellence in Topological Quantum Computation, University of Chinese Academy of Sciences, Beijing 100190, China}

\date{\today}

\begin{abstract}
We investigate the fundamental constraint on amplifying the asymmetry in quantum states with correlating catalysts. 
Here a correlating catalyst is a finite-dimensional auxiliary, which exactly preserves its reduced state while allowed to become correlated with the quantum system.
Interestingly, we prove that under translationally invariant operations, uncorrelating catalysts in pure states are useless in any state transformation, while with a correlating catalyst , one can extend the set of accessible states from an initially asymmetric state.
Moreover, we show that the power of a catalyst increases with its dimension, and further, with a large enough catalyst, a qubit state with an arbitrarily small amount of asymmetry can be converted to any mixed qubit state.
In doing so, we build a bridge between two important results concerning the restrictions on coherence conversion, the no-broadcasting theorem and the catalytic coherence.
Our results may also apply to the constraints on coherence evolution in quantum thermodynamics and to the distribution of timing information between quantum clocks.
\end{abstract}

\maketitle

\section{Introduction}
Finding out whether a quantum state can be converted to another under a set of restricted operations is a problem originating from the entanglement theory \cite{PhysRevLett.83.436}, and has recently been studied in a variety of resource theories \cite{Horodecki2013thermalmajor,PhysRevLett.117.030401,RevModPhys.91.025001,Marvian2014}. Moreover, in resource theories such as entanglement \cite{PhysRevLett.83.3566}, athermality \cite{Brandao3275}, coherence \cite{PhysRevA.93.042326,PhysRevA.100.042323,Vaccaro_2018,PhysRevLett.119.140402}, and quantum randomness \cite{PhysRevX.8.041016}, catalysts are employed to enhance the ability of state conversion. 
A catalyst is an ancilla which interacts with the system and then returns to the exact original state. Conventionally, the catalyst is required to be uncorrelated with the system after the process \cite{RevModPhys.91.025001,Brandao3275,Turgut,PhysRevLett.121.190504}. 
Nevertheless, recent studies suggest that the creation of correlations may greatly extend the set of accessible states \cite{PhysRevX.8.041051,PhysRevLett.115.150402,nc2019_corr_thermo,PhysRevE.99.042135,PhysRevLett.122.210402}. 
In particular, in the resource theory of athermality, if this uncorrelation requirement is lifted, the catalyst becomes more powerful.
Namely, it enables state conversions which are not achievable using an uncorrelating catalyst \cite{PhysRevX.8.041051}. Moreover, in resource theories governed by majorization (such as athermality), any resourceful state, supplied with sufficiently many copies, can be used as a catalyst for any allowed transformation \cite{lipkabartosik2020states}.

The superposition between different eigenstates of a conserved observable is a valuable resource of asymmetry in many tasks such as quantum metrology \cite{PhysRevLett.96.010401,PhysRevA.93.042107}, quantum clocks \cite{PhysRevLett.82.2207}, and quantum thermodynamics \cite{PhysRevLett.120.150602}. 
Also, the constraints on the asymmetry dynamics under covariant operations impose essential limitations on thermodynamics processes beyond free energy \cite{Lostaglio15nc}, and on quantum speed limits \cite{PhysRevX.6.021031}.
Surprisingly, evidence has been uncovered that catalysts might be useless in the resource theory of asymmetry. The no-catalysis theorem, as proved in Ref. \cite{Marvian_2013}, states that if a pure state cannot be converted to another pure state using operations which are symmetric under a compact Lie group, then any finite-dimensional catalyst in a pure state can not enable this conversion. 
Furthermore, by the no-broadcasting theorem of asymmetry \cite{PhysRevLett.123.020403,PhysRevLett.123.020404}, the creation of asymmetry in an initially symmetric state is impossible even with a correlating finite-dimensional catalyst, in comparison to the protocol of catalytic coherence \cite{PhysRevLett.113.150402}, where an arbitrary amount of coherence between energy levels can be created by interacting the system with an infinite-dimensional catalyst. In order to explore the crossover between the no-broadcasting theorem and the catalytic coherence, we ask the following question: To what extent can a finite-dimensional catalyst enlarge the set of accessible states under symmetric transformations?

Because the creation of correlations between the system and the catalyst may ease the state transformation, and this correlation does not affect the power of the catalyst in other state transformations, we allow the catalyst to become correlated with the system on which it acts. 
In this paper, we impose three restrictions on the catalytic system: (1) it is finite-dimensional; (2) its reduced state is exactly identical before and after the state conversion; and (3) it is uncorrelated with the system \emph{before} the state transformation. 
Here we first prove a general result that a catalyst can extend the set of accessible states only if it is in a mixed state,
which generalizes the no-catalysis theorem in Ref. \cite{Marvian_2013}. 
Then we show that, in contrast to other resource theories, there is no bound on amplifying the asymmetry with correlating catalysts. 

That is, any qubit state with an arbitrarily small amount of asymmetry can be converted to a state arbitrarily close to the state with maximal asymmetry, as long as the dimension of the catalyst is large enough. 
The applications of our results to the constraints on catalytic coherence evolution in quantum thermodynamics, and to the distribution of timing information, are also discussed.

\section{Notions}
In the resource theory of time-translation asymmetry, the free states are the symmetric states $\rho_\mathrm{sym}$, which are invariant during the evolution under Hamiltonian $H$, i.e., $e^{-iHt}\rho_\mathrm{sym}e^{iHt}=\rho_\mathrm{sym}$, $\forall t$ (equivalently, $[\rho_\mathrm{sym},H]=0$). 
The free operations are the translationally invariant operations (TIOs), or covariant operations, which are defined as completely-positive and trace-preserving maps $\cptp$ satisfying $\cptp\left(e^{-iHt}\rho e^{iHt}\right)=e^{-iHt}\cptp(\rho) e^{iHt},\ \forall \rho,t$. The set of states that can be converted to from a given state $\rho$ under TIOs is called the TIO cone of $\rho$, labeled $\tiocone(\rho)\equiv\{\rho':\rho'=\cptp(\rho),\cptp\in\tio\}$.

The correlating-catalytic TIO (CCTIO) in a system $S$ with Hamiltonian $H_S$ is implemented by coupling $S$ to a finite-dimensional auxiliary $C$ with Hamiltonian $H_C$ via a  global translationally invariant operation which preserves the reduced state of $C$. 
Here by ``correlating,'' we mean that the catalyst $C$ is initially uncorrelated with $S$ but allowed to become correlated with $S$ in the output.
Precisely, we say that a state $\rho$ can be transformed to $\rho'$ by CCTIO, if there exists a finite-dimensional auxiliary system in state $\sigma_C$, and a global TIO $\cptp$ satisfying $\cptp\left(e^{-i(H_S+H_C)t}\cdot e^{i(H_S+H_C)t}\right)=e^{-i(H_S+H_C)t}\cptp(\cdot) e^{i(H_S+H_C)t}$ such that
\begin{equation}
\cptp(\rho\otimes\sigma)=\rho'|\sigma.\label{eq:cctio}
\end{equation}
Here the label $\rho'|\sigma$ means a bipartite state of $S$ and $C$, whose reduced states are $\rho'$ on $S$ and $\sigma$ on $C$. Notably, the state of $C$ is identical before and after the action of $\cptp$. The set of states achievable under CCTIO from $\rho$ is called the CCTIO cone of $\rho$, and the auxiliary $C$ is called the correlating catalyst. When the dimension of $C$ is restricted to $d$, the CCTIO cone of $\rho$ is labeled $\cccone^{(d)}(\rho)$.

\section{Catalysts in pure states are useless}
When the catalyst is in a pure state $\sigma=\ket{\phi}\bra{\phi}$, the transformation in Eq. (\ref{eq:cctio}) reads
\begin{equation}
\cptp(\rho\otimes\ket{\phi}\bra{\phi})=\rho'\otimes\ket{\phi}\bra{\phi}.\label{eq:ctio}
\end{equation}
Because it is required that the state of the catalyst is exactly retained, the purity of $\sigma$ ensures that $S$ and $C$ are not correlated in the output. For an asymmetry monotone $I(\rho)$ which is additive on tensor products 
(such as the quantum Fisher information \cite{clock2003} and Wigner-Yanase skew information \cite{Marvian12}; see Appendix \ref{app:measure}), Eq. (\ref{eq:ctio}) implies that $I(\rho')+I(\ket{\phi}\bra{\phi})=I(\rho'\otimes\ket{\phi}\bra{\phi})\leq I(\rho\otimes\ket{\phi}\bra{\phi})=I(\rho)+I(\ket{\phi}\bra{\phi})$, and hence, $I(\rho')\leq I(\rho)$. This means that with a finite-dimensional pure catalytic state, the asymmetry monotone $I$ can never be increased. 

Yet, it is not as simple to see whether other asymmetry monotones, which are not additive on tensor products, are also monotonic under the catalyzed transformation in Eq. (\ref{eq:ctio}). 
In the following theorem, we show a stronger result. Namely, any catalyst in a pure state cannot enable state transformations which are not achievable by TIO. The proof of this theorem is given in Appendix \ref{app:th1}.

\begin{theorem}\label{th1}
If $\rho$ cannot be transformed to $\rho'$ under TIO, then the transformation $\rho\otimes\ket{\phi}\bra{\phi}\mapsto\rho'\otimes\ket{\phi}\bra{\phi}$ under TIO is also not achievable for any choice of pure state $\ket{\phi}$.
\end{theorem}

This theorem is our first main result. It generalizes the no-catalysis theorem \cite{Marvian_2013}, in which the states $\rho$ and $\rho'$ were restricted to pure states, and indicates that catalysts in pure states are useless in any state transformation under covariant operations.

\section{Extending the set of accessible states with correlating catalysts}
When the catalyst is in a mixed state $\sigma$, it may become correlated with the system after the transformation. 
Because there are correlated states $\rho'|\sigma$ such that $I(\rho'|\sigma)<I(\rho')+I(\sigma)$ \cite{PhysRevLett.123.020403,PhysRevLett.123.020404,Takagi2019}, it is possible that $I(\rho')>I(\rho)$, i.e., a correlating catalyst may enable state transformations which are not achievable under TIOs.

In Ref. \cite{PhysRevE.99.042135}, a stationary machine was designed to control and amplify the energetic coherence in quantum systems. 
It gives evidence that, with the help of a correlating catalyst (the stationary machine), one can achieve state transformations which are not realizable via TIOs. 
Nevertheless, it is not quite straightforward to see whether the state of the catalyst (machine) is \emph{exactly} identical after each round, due to the approximations in deriving the master equations.

Here we give an analytic example, which shows that a global TIO acting on $S$ and $C$ can transform the state of $S$ to a state not achievable under TIOs, while strictly preserving the reduced state of $C$. 
Consider a system qubit with Hamiltonian $H_S=\frac{\Delta}{2}\sigma^z$ and a catalyst qubit with Hamiltonian $H_C=\frac{\Delta}{2}\sigma^z$, where $\sigma^{x,y,z}$ denote the Pauli matrices and $\Delta>0$ is the energy gap. 
Initially, the two-qubit state of $SC$ reads $\rho(\eta)\otimes\sigma^{\uparrow}(\eta)$, where
\begin{eqnarray}
\label{eq:rhoeta}\rho(\eta)&=&\frac12(\iden+\eta\sigma^x),\ 0<\eta<1,\\
\label{eq:cueta}\sigma^\uparrow(\eta)&=&\frac12\iden+\frac{\sqrt3\eta}{4}\sigma^x+\frac{4-\eta^2}{12}\sigma^z,
\end{eqnarray}
are states of $S$ and $C$ respectively.
After the application of the global covariant operation $\cptp^\uparrow(\cdot)=K_0(\cdot)K_0^\dagger+K_1(\cdot)K_1^\dagger$ with
\begin{equation}
K_0=\begin{pmatrix}
1& 0& 0&0\\
0 & \frac14 & \frac{\sqrt3}{4} & 0\\
0 & \frac{\sqrt3}{4} & \frac34 &0 \\
0 &0 &0 & 1
\end{pmatrix},
K_1=\begin{pmatrix}
0& -\frac{\sqrt3}{2} & \frac12 & 0\\
0 & 0 &0  & 0\\
0 & 0 &0  & 0\\
 0& 0& 0& 0
\end{pmatrix},\label{eq:kraus}
\end{equation}
we obtain a correlated two-qubit state, whose reduced states are 
\begin{equation}
    \rho^\uparrow(\eta)=\frac12\iden+\frac{25\eta-\eta^3}{48}\sigma^x+\frac{1-\eta^2}{6}\sigma^z
\end{equation}
on $S$ and $\sigma^{\uparrow}(\eta)$ as in Eq. (\ref{eq:cueta}) on $C$.

The TIO cone of $\rho(\eta)$ reads $\tiocone[\rho(\eta)]=\{\xi|\xi=\frac12[\iden+r(\cos\phi\sigma^x+\sin\phi\sigma^y)+r_z\sigma^z],r_z\in[-1,1],\phi\in[0,2\pi),0\leq r\leq\eta\sqrt{1-|r_z|}\}$ 
(see Appendix \ref{app:cone} for details).
Hence for any state $\xi\in \tiocone[\rho(\eta)]$, it holds that $\Tr(\xi\sigma^x)\leq\eta$. Because $\Tr[\rho^\uparrow(\eta)\sigma^x]>\eta$, we have $\rho^\uparrow(\eta)\notin\tiocone[\rho(\eta)]$. 
This means that the state transformation from $\rho(\eta)$ to $\rho^\uparrow(\eta)$, which is not achievable by TIO, can be enabled by employing a correlated catalytic qubit. 
Notably, the reduced state $\sigma^\uparrow(\eta)$ of $C$ is \emph{exactly} identical before and after the action of the global TIO $\cptp^\uparrow$, which excludes the phenomenon of embezzlement.

\begin{figure}[htbp]
        \centering
        \includegraphics[width=0.45\textwidth]{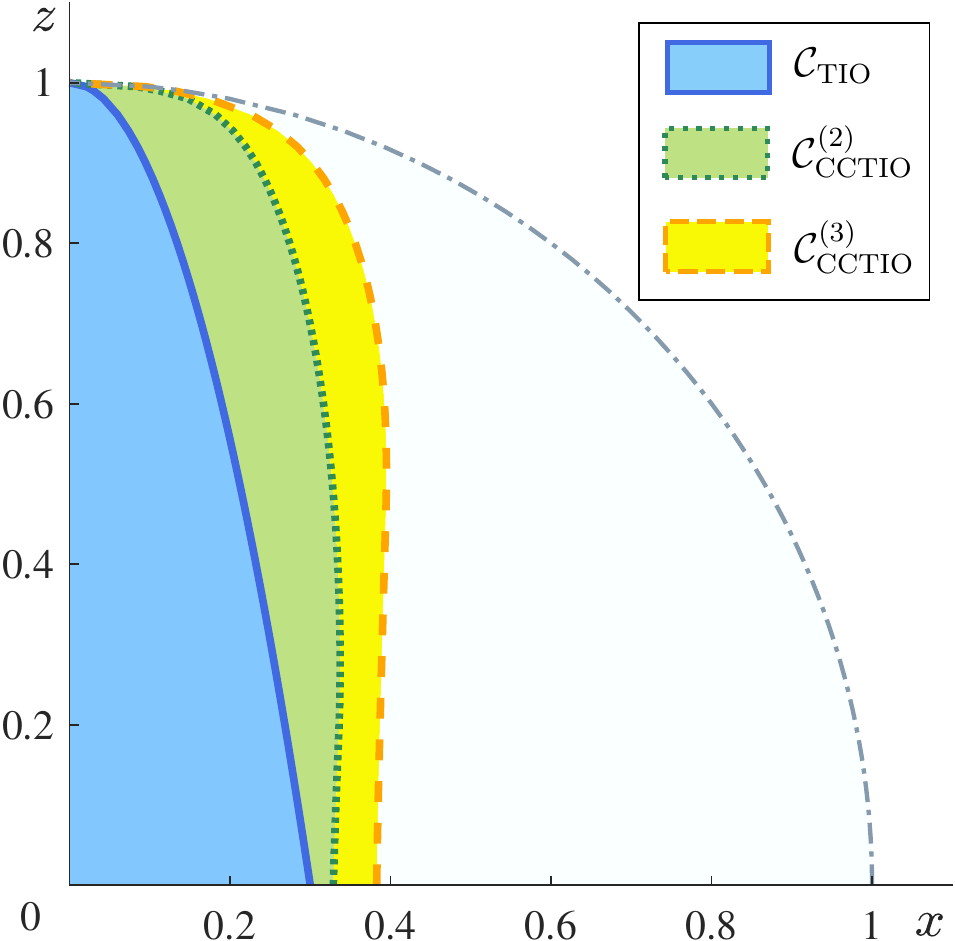} 
        \caption{Comparison of the TIO cone, and the CCTIO cone with two-dimensional and three-dimensional catalysts, of a qubit state. The initial state is $\rho(\eta)$ as in Eq. (\ref{eq:rhoeta}) with $\eta=0.3$.} 
        \label{fig:CCTIO}
\end{figure}

Next we analyze whether the power of catalysis is affected by the dimension $d$ of $C$. 
Here we focus on the whole set of accessible states from a given initial state, instead of the amplification of some asymmetry measure, in order to avoid the dependence on the choice of asymmetry measure.
Specifically, we numerically calculate the CCTIO cone with $d=2,3$ of state $\rho(\eta)$ 
(see Appendix \ref{app:numerical} for details).
In the Bloch presentation, the basic structure of the CCTIO cone of a qubit state is that it is rotationally symmetric about the $z$ axis. 
For state $\rho(\eta)$, it is also symmetric about the $xy$ plane 
(see Appendix \ref{app:cone} for details).
In Fig. \ref{fig:CCTIO}, we plot $\tiocone[\rho(\eta)]$, $\cccone^{(2)}[\rho(\eta)]$, and $\cccone^{(3)}[\rho(\eta)]$ within the $xz$ plane with $x\geq0$ and $z\geq0$. 
Clearly, $\cccone^{(3)}$ is strictly larger than $\cccone^{(2)}$. 
This means that, with a correlating catalyst of higher dimension, more state conversions can be realized.

\section{Unbounded amplification of asymmetry}
Here we propose a protocol to show that, any input qubit state $\rho$ satisfying $[\rho,H_S]\neq0$ can be transformed to a state arbitrarily close to $\rho^+=\frac12(\iden+\sigma^x)$ via CCTIO, given that the dimension of the catalyst is large enough (but still finite). 
Note that $\rho^+$ is the qubit state with the maximum amount of asymmetry, when a variety of asymmetry measures, including those based on the skew information, relative entropy, robustness, etc., are employed 
(see Appendix \ref{app:measure} for details).

This protocol is a generalization of the analytic example in the last section. 
Here the catalyst $C$ consists of $N$ particles, labeled $C_1,\dots,C_i,\dots,C_N$, each of which contains two qubits, $C_{i1}$ and $C_{i2}$. 
The Hamiltonian of the $i$th particle $C_i$ reads $H_{C_i}=H_{C_{i1}}+H_{C_{i2}}$ with $H_{C_{i1}}=\frac{\Delta}{2}\sigma^z$ and  $H_{C_{i2}}=\frac{\Delta'}{2}\sigma^z$.
The reduced state of $C_i$ is set in the form
\begin{equation}\label{eq:sigma_eta_i}
\sigma(\eta_i)=\frac12\sigma^\uparrow(\eta_i)\otimes\ket{\uparrow}\bra{\uparrow}+\frac12\sigma^\downarrow(\eta_i)\otimes\ket{\downarrow}\bra{\downarrow},
\end{equation}
where $\sigma^\uparrow(\eta_i)$ as in Eq. (\ref{eq:cueta}) and $\sigma^\downarrow(\eta_i)\equiv\sigma^x\sigma^\uparrow(\eta_i)\sigma^x$ are states of qubit $C_{i1}$, the parameter $\eta_i$ depends on the state of the system qubit, and $\ket{\uparrow}$ and $\ket{\downarrow}$ are energy eigenstates of qubit $C_{i2}$.

To start with, we convert the initial asymmetric state $\rho$ of the system qubit to $\rho(\eta_1)$ with $\eta_1>0$ via a local TIO $\cptp_0$. Then the system qubit is coupled to each of the $N$ particles via TIO consequently. In each turn, the operation acting on $SC_i$ reads
\begin{equation}
\cptp_i=\cptp^{\uparrow}\otimes{\Pi^\uparrow}+ \cptp^{\downarrow}\otimes{\Pi^\downarrow},
\end{equation}
where the two-qubit TIOs $\cptp^{\uparrow}$ and $\cptp^{\downarrow}$ are applied to $SC_{i1}$, and $\Pi^\uparrow$ and $\Pi^\downarrow$ are projectors to the energy eigenstates of qubit $C_{i2}$. 
Here $\cptp^\uparrow$ is in the form of Eq. (\ref{eq:kraus}), and $\cptp^\downarrow\equiv \mathcal U^x\circ\cptp^\uparrow\circ\mathcal U^x$, with $\mathcal U^x(\cdot)=\sigma^x\otimes\sigma^x(\cdot)\sigma^x\otimes\sigma^x$ and $\circ$ denoting the composition of quantum operations.
It can be checked by definition that $\cptp_i$ is a TIO on the composite system $SC_{i}$.
Effectively, we have
\begin{eqnarray}
&&\cptp_i[\rho(\eta_i)\otimes\sigma(\eta_i)]\nonumber\\
&=&\frac12\rho^\uparrow(\eta_i)|\sigma^\uparrow(\eta_i)\otimes\ket{\uparrow}\bra{\uparrow}+\frac12\rho^\downarrow(\eta_i)|\sigma^\downarrow(\eta_i)\otimes\ket{\downarrow}\bra{\downarrow}\nonumber\\
&=&\rho(\eta_{i+1})|\sigma(\eta_i),\label{eq:cptpi}
\end{eqnarray}
where $\rho^\downarrow(\eta)\equiv\sigma^x\rho^\uparrow(\eta)\sigma^x$, and $\eta_{i+1}=\eta_i+\frac{1}{24}\eta_i(1-\eta_i^2)$ for $i=1,\dots,N$.
This means that after the action of $\cptp_i$, the parameter $\eta$ in the state of $S$ is increased by $\Delta\eta_i\equiv\eta_{i+1}-\eta_{i}=\frac{1}{24}\eta_i(1-\eta_i^2)$,
which is strictly positive with $0<\eta_i<1$.
Therefore, as long as the initial state $\rho$ is asymmetric (which ensures $\eta_1>0$), we can achieve the state $\rho(\eta_{N+1})$ with $\eta_{N+1}\rightarrow1$ for finite $N$.
In other words, the state conversion from a state with an arbitrarily small amount of asymmetry to a state arbitrarily close to $\rho^+$ can be achieved by our protocol, if the catalyst is large enough (but still finite).

Now let us illustrate that the operation involved in our protocol is CCTIO. 
Because $\cptp_i$ are TIOs, the composite operation $\cptp\equiv\cptp_N\circ\cdots\circ\cptp_1\circ\cptp_0$ is also a TIO. 
In the following, we construct the $N$-particle catalytic state $\sigma$ such that it is not affected by $\cptp$, i.e., $\Tr_S[\cptp(\rho\otimes\sigma)]=\sigma$.
First, the reduced state of particle $C_i$ is not affected either by $\cptp_i$ due to Eq. (\ref{eq:cptpi}) nor by $\cptp_{i'}$ with $i'\neq i$, so we have $\Tr_{\backslash i}[\cptp(\rho\otimes\Sigma)]=\sigma(\eta_i),\ \forall i$, for any $N$-particle state $\Sigma$ satisfying $\Tr_{\backslash i}(\Sigma)=\sigma(\eta_i),\forall i$. 
Here $\Tr_{\backslash i}$ means a partial trace on all systems except $C_i$. 
Then let $\sigma^{(1)}=\sigma\otimes\cdots\otimes\sigma(\eta_N)$ and $\sigma^{(j+1)}=\Tr_{S}[\cptp(\rho\otimes\sigma^{(j)})]$, and we have $\Tr_{\backslash i}(\sigma^{(j)})=\sigma(\eta_i),\forall i,j$. Now we define $\sigma=\lim_{n\rightarrow\infty}\frac1n\sum_{j=1}^n\sigma^{(j)}$, and then
\begin{eqnarray}
&&\Tr_S[\cptp(\rho\otimes\sigma)]=\lim_{n\rightarrow\infty}\frac1n\sum_{j=0}^n\Tr_S[\cptp(\rho\otimes\sigma^{(j)})]\nonumber\\
&=&\sigma+\lim_{n\rightarrow\infty}\frac1n[\sigma^{(n+1)}-\sigma^{(1)}]=\sigma.
\end{eqnarray}
In total, we have $\cptp(\rho\otimes\sigma)=\rho(\eta_{N+1})|\sigma$. This means that $\rho$ is converted to $\rho(\eta_{N+1})$ via CCTIO.

The condition $[\rho,H_S]\neq0$ on the input state $\rho$ is essential in our protocol. 
If it is satisfied, we have $\rho(\eta_1)\in\tiocone(\rho)$ with $\eta_1>0$ 
(see Appendix \ref{app:cone} for details).
Otherwise, $\eta_1=0$, and hence, $\Delta\eta_i=0,\forall i$. 
This means that, if we start from a symmetric state, it remains symmetric for arbitrarily large $N$. 
This is compatible with the no-broadcasting theorem of asymmetry \cite{PhysRevLett.123.020403,PhysRevLett.123.020404}. Namely, by employing a correlating catalyst, one can amplify the asymmetry in an asymmetric state, instead of creating asymmetry in a symmetric state.

Here we mention that, in general, it is impossible to reach $\rho^+$ \emph{exactly} with any finite-dimensional catalyst. Suppose there is a finite-dimensional catalytic state $\sigma_C$ and a global TIO $\cptp$ such that $\cptp(\rho\otimes\sigma_C)=\rho^+|\sigma_C$. Because $\rho^+$ is a pure state, the bipartite state in the output is not correlated, i.e., $\rho^+|\sigma_C=\rho^+\otimes\sigma_C$. It follows that $I(\rho)\geq I(\rho^+)$, which holds only if $\rho$ is on the equator of the Bloch sphere. Therefore, it is impossible to transform any state (other than the ones equivalent to $\rho^+$ by symmetric unitaries) exactly to $\rho^+$ by any CCTIO process.

 Nevertheless, because $\tiocone(\rho^+)$ does not include all the qubit states, it is not straightforward to see whether one can transform $\rho$ to any qubit state approximately via CCTIO. In the following, we generalize the above protocol and give an affirmative answer to this question.

\begin{theorem}\label{th2}
For any pair of qubit states $\rho$ and $\rho'$, the state conversion from $\rho_\epsilon$ to $\rho'_{\epsilon'}$ is achievable under CCTIO.
\end{theorem}

\begin{proof}
    In the generalized protocol, we start from a state $\rho(\eta,z)=\frac12(\iden+\eta\sigma^x+z\sigma^z)$ with $\eta>0$, and transform it to the state
\begin{equation}\label{eq:rho_prime}
    \rho'(\eta',z')=\frac12(\iden+\eta'\sigma^x+z'\sigma^z)
\end{equation}
with $0\leq\eta^{\prime}<\sqrt{1-z^{\prime2}}$, under CCTIO.

Most of the setup of the generalized protocol is the same as the original one. Here we make two changes. First, at the beginning of the protocol, we transform the input state $\rho(\eta,z)$ to $\rho(\eta_1,z')$ [instead of $\rho(\eta_1)$ that $\Tr\left[\rho(\eta_1)\sigma^z\right]=0$] with $\eta_1>0$ via local TIO. The second change is that, the reduced state of each two-qubit particle $C_i$ in the catalyst is set to be
\begin{equation}
\begin{aligned}
\sigma(\eta_i,z')
\equiv &(\frac{1}{2}+\frac{z'}{3})\sigma^{\uparrow}(\eta_i,z')\otimes\ket{\uparrow}\bra{\uparrow}\\
+&(\frac{1}{2}-\frac{z'}{3})\sigma^\downarrow(\eta_i,z')\otimes\ket{\downarrow}\bra{\downarrow},
\end{aligned}
\end{equation}
where
\begin{equation}
    \begin{aligned}
    \sigma^{\uparrow}(\eta_i,z')&\equiv\frac{1}{2}\iden+\frac{\sqrt{3}\eta_i}{2(2+z')}\sigma^x+\frac{z' (8 + 3 z')+4-\eta_i^2}{2(2+z')(3+2z')}\sigma^z,\\
    \sigma^\downarrow(\eta_i,z')&\equiv\frac{1}{2}\iden+\frac{\sqrt{3}\eta_i}{2(2-z')}\sigma^x+\frac{z' (8 - 3 z')-4+\eta_i^2}{2(2-z')(3-2z')}\sigma^z
    \end{aligned}
\end{equation}
are states of qubit $C_{i1}$, and $\ket{\uparrow}$ and $\ket{\downarrow}$ are energy eigenstates of qubit $C_{i2}$. Here the parameter $z'$ is determined by the target state $\rho(\eta',z')$, and $\eta_i$ depend on both the initial state and the target state.

After the action of $\cptp_i$, the state of the system qubit becomes $\rho(\eta_{i+1},z')$ with
\begin{equation}
    \eta_{i+1}=\eta_i+ \frac{\eta_i(1-z^{\prime 2}-\eta_i^2)}{6(4-z^{\prime 2})},
\end{equation}
for $i=1,\dots,N$. It is directly checked that $\Delta\eta_i\equiv \eta_{i+1}-\eta_i>0$ for $0<\eta_i<\sqrt{1-z^{\prime2}}$. 
Therefore, with finite $N$, we can achieve state $\rho(\eta_N,z')$ with $\eta'<\eta_N<\sqrt{1-z^{\prime2}}$. 
Clearly, $\rho(\eta',z')\in\tiocone[\rho(\eta_N,z')]$. 
Hence by our protocol, any state as in Eq. (\ref{eq:rho_prime}) can be obtained via CCTIO from state $\rho(\eta,z)$ with $\eta>0$.

Note that in the Bloch representation, any mixed state in the $xz$ plane with $x\geq0$ can be presented in the form of Eq. (\ref{eq:rho_prime}). By the rotational symmetry of the CCTIO cone of a qubit state, we conclude that the CCTIO cone of an asymmetric state contains all the mixed qubit states. This completes the proof.
\end{proof}

\section{Application}

\subsection{Application to the resource theory of athermality}
In the resource theory of athermality \cite{Lostaglio_2019}, the free operations are the thermal operations, which can be implemented by coupling the system to a reservoir at inverse temperature $\beta$ and then shutting down the interaction after a while. 
Here the interaction preserves the total energy, which ensures that thermal operations are translationally invariant. 

It was shown in Ref. \cite{PhysRevX.8.041051} that, for any pair of symmetric states $\rho$ and $\rho'$, there exists a finite-dimensional catalytic system in state $\sigma$ and a global thermal operation $\cptp_\mathrm{TO}$ such that $\cptp_\mathrm{TO}(\rho\otimes\sigma)=\rho'_\epsilon|\sigma$, 
where $\rho'_\epsilon$ is arbitrarily close to $\rho'$, if and only if $F(\rho)\geq F(\rho')$. Here the free energy $F(\rho):=\Tr(\rho H)-\frac{\ln2}{\beta}S(\rho)$, with the von Neumann entropy $S(\rho):=-\Tr(\rho\log_2\rho)$. However, this statement cannot be generalized to the fully quantum regime where $\rho'$ is asymmetric but $\rho$ may be symmetric, because of the no-broadcasting theorem of asymmetry \cite{PhysRevLett.123.020404}.

From Theorem \ref{th2}, the restriction of no-broadcasting can be lifted if the initial state has a small amount of asymmetry. Therefore, the following conjecture may hold. For a given pair of two quantum states $\rho$ and $\rho'$, there exists a finite-dimensional system in state $\sigma$ and a global thermal operation $\cptp_\mathrm{TO}$ such that
\begin{equation}
    \cptp_\mathrm{TO}(\rho_\epsilon\otimes\sigma)=\rho'_{\epsilon'}|\sigma,
\end{equation}
if and only if $F(\rho)\geq F(\rho')$. We leave further discussion of this conjecture to future work.

\subsection{Application to the quantum clock}
After a system is prepared in an asymmetric state $\rho$, it evolves according to its Hamiltonian $H$ as $\rho(t)=e^{-iHt}\rho e^{iHt}$. 
From the asymmetric condition $[\rho,H]\neq0$, the states $\rho(t)$ are not all the same and thus contain some time information. 
In this sense, the evolution of an asymmetric state is considered as the pointer of a quantum clock \cite{clock2003}. 
In general, a quantum clock is identified by the pair $(\rho,H)$, and its accuracy is quantified by the Fisher timing information $I_F(\rho,H):=\Tr(\dot{\rho}\Delta^{-1}_\rho\dot{\rho})$, 
where $\dot{\rho}:=i[\rho,H]$ and $\Delta_\rho B:=\frac12(\rho B+B\rho)$.

A previous result has shown that \cite{PhysRevLett.123.020404}, it is impossible to distribute the timing information of a quantum clock into a system with zero timing information, without affecting the quantum state of the clock. Formally, let $(\rho_1,H_1)$ be a finite-dimensional quantum clock $S_1$, and the pair $(\rho_2,H_2)$ with $[\rho_2,H_2]=0$ denote a system $S_2$ without timing information. Then the no-broadcasting theorem of asymmetry implies that there does not exist a global covariant operation $\cptp$ such that $\cptp(\rho_1\otimes\rho_2)=\rho_1|\rho'_2$ and $[\rho'_2,H_2]\neq0$.

In some realistic circumstances, completely dephasing operations are difficult to implement \emph{exactly} \cite{Liu2018,PhysRevLett.82.5181}. 
Therefore, when initializing the second system $S_2$ to satisfy $[\rho_2,H_2]=0$, one might obtain a pair $(\rho_2,H_2)$ with arbitrarily small but positive Fisher timing information $I_F(\rho_2,H_2)=\epsilon>0$. 
Theorem \ref{th1} indicates that, if $\rho_1$ is pure, the timing information in $S_2$ is still negligible after the action of a global covariant operation that preserves $\rho_1$ and, thus, generalizes the no-distributing principle of timing information to this noisy case.

Nevertheless, if $\rho_1$ is mixed, then it is possible to make the clock $S_2$ more accurate without affecting the state of the clock $S_1$. 
Still, it should be noted that, the design of $S_1$ depends heavily on the initial state $\rho_2$. Therefore, in order to deterministically amplify the accuracy of $S_2$, one has to know the exact form of $\rho_2$.

\section{Discussion and conclusion}
We have investigated the ability of a correlating catalyst to amplify the asymmetry of a quantum system. 
While a catalyst in a pure state cannot extend the set of accessible states from any input state, a large enough catalyst in a mixed state can enable the conversion from a qubit state, which is arbitrarily close to (but not in) the set of symmetric states, to a state arbitrarily close to $\rho^+$, which is the qubit state with maximum asymmetry. 
The asymmetry in the initial state is essential due to the no-broadcasting theorem \cite{PhysRevLett.123.020403,PhysRevLett.123.020404}. 
Besides, in the limit of infinite-dimensional catalysis, $\rho^+$ can be reached as in catalytic coherence \cite{PhysRevLett.113.150402}. 
Hence, our result bridges theses two important results concerning the restrictions on the coherence dynamics under translationally invariant operations.

It is also of interest to study the amplification of asymmetry in higher dimensional systems. 
The main difficulty in solving this problem is that it is numerically hard to calculate the cones of a high-dimensional state. 
An alternative way of dealing with this problem is to calculate the amount of asymmetry amplified by correlating catalysts, but such results would depend on the choice of the asymmetry measure. 
A more meaningful question to ask is, Can we generalize Theorem \ref{th2} to any finite-dimensional system? We conjecture that the answer is ``Yes.'' 
One clue is to use the elementary framework as in Ref. \cite{Lostaglio2018elementarythermal}, i.e., to operate on two energy levels at a time.
This may be feasible, because the bounds on coherence dynamics under TIOs \cite{PhysRevX.5.021001,PhysRevLett.115.210403} do not eliminate the possibility of manipulating the coherence between two energy levels while preserving the coherence between other energy levels.

We have seen that, in resource theories of athermality \cite{PhysRevX.8.041051} and asymmetry, correlating catalysts are strictly more powerful than uncorrelated ones. 
A related open problem is as follows. 
Is there any resource theory, in which the creation of correlations between the catalyst and the system can never extend the set of accessible states? 
A sufficient condition is that for any given bipartite state (whose form is known) of two resourceful systems, one can decouple the two systems while preserving their reduced states by free operations. 
This free-decoupling condition is potentially an interesting problem on its own and, to our knowledge, has been discussed only in the resource theory of athermality \cite{PhysRevA.99.012104}. 
If a resource theory satisfies the free-decoupling condition, then any resource monotones are super-additive. 
However, the converse is not obvious: a variety of coherence monotones are proved to be super-additive \cite{PhysRevX.6.041028,PhysRevA.95.042328,RevModPhys.89.041003,HU20181}, but it is not straightforward to prove that the resource theory of coherence satisfies the free-decoupling condition.

\begin{acknowledgments}
This work was Supported by NATIONAL NATURAL SCIENCE FOUNDATION OF CHINA under Grant No. 11774205, and Young Scholars Program of Shandong University.
\end{acknowledgments}


\appendix

\section{PROOF OF THEOREM \ref{th1}}\label{app:th1}
Here we first prove a series of lemmas and then present the proof of Theorem \ref{th1}.

\begin{lemma}\label{lemma:neother}
(Neother's theorem \cite{Marvian_2013}). For two pure states $\ket{\psi}$ and $\ket{\psi'}$, there exists a covariant unitary $V$ such that $V\ket{\psi}=\ket{\psi'}$ if and only if
\begin{equation}
\bra{\psi}e^{-iHt}\ket{\psi}=\bra{\psi'}e^{-iHt}\ket{\psi'},\ \forall t.
\end{equation}
\end{lemma}

\begin{lemma}\label{lemma:equivalent}
(Marvian and Spekkens \cite{Marvian_2013}). For two pure states $\ket{\psi},\ket{\psi'}\in \mathcal H_S$ and a pure catalytic state $\ket{\phi}\in \mathcal H_C$, if a global covariant unitary $U$ induces the transformation
\begin{equation}
\label{eq:equi} U(\ket{\psi}\otimes\ket{\phi})=\ket{\psi'}\otimes\ket{\phi},
\end{equation}
then there exists a covariant unitary $V$ acting on $\mathcal H_S$ such that $V\ket{\psi}=\ket{\psi'}$.
\begin{proof}
From Lemma \ref{lemma:neother} and Eq. (\ref{eq:equi}), we have $\bra{\psi}e^{-iH_St}\ket{\psi}\bra{\phi}e^{-iH_Ct}\ket{\phi}=\bra{\psi'}e^{-iH_St}\ket{\psi'}\bra{\phi}e^{-iH_Ct}\ket{\phi},\ \forall t$. Because $\bra{\phi}e^{-iH_Ct}\ket{\phi}\neq0$, it follows that $\bra{\psi}e^{-iH_St}\ket{\psi}=\bra{\psi'}e^{-iH_St}\ket{\psi'},\ \forall t$. Then from Lemma \ref{lemma:neother}, $\ket{\psi}$ can be transformed to $\ket{\psi'}$ via a covariant unitary.
\end{proof}
\end{lemma}

\begin{lemma}\label{lemma:pre}
Let $\ket{\tilde a_0},\dots,\ket{\tilde a_{n-1}},\ket{\tilde b}$ be state vectors in Hilbert space $\mathcal H$ which need not be normalized. If for a given unitary $V_1$ the following equations hold
\begin{equation}
\bra{\tilde a_j}\ket{\tilde b}=\bra{\tilde a_j}V_1\ket{\tilde b},\ \forall j,\label{eq:inner}
\end{equation}
then there exists a unitary $V$ such that $V\ket{\tilde a_j}=\ket{\tilde a_j},\ \forall j$ and $V\ket{\tilde b}=V_1\ket{\tilde b}$.
\begin{proof}
Here we first consider two trivial cases. \\
Case 1. $\ket{\tilde b}=0$. In this case we simply set $V=\iden$.\\
Case 2. $\ket{\tilde a_j}=0,\ \forall j$. In this case, one can set $V=V_1$.

In the situation where the above two cases are excluded, let $\{\ket{a_0},\dots,\ket{a_{d-1}}\}$ be an orthonormal basis for the subspace $\mathcal A\equiv\SPAN\{\ket{\tilde a_0},\dots,\ket{\tilde a_{n-1}}\}$, and $\mathcal A_\perp\subset\mathcal H$ be the subspace orthogonal to $\mathcal A$. Then $\ket{\tilde b}$ can be written as
\begin{equation}
\ket{\tilde b}=\sum_{k=0}^{d-1}b_k\ket{a_k}+b_d\ket{b_\perp},
\end{equation}
where $\ket{b_\perp}\in\mathcal A_\perp$. From Eq. (\ref{eq:inner}), for any state $\ket{a}\in\mathcal A$, it holds that $\bra{ a}V_1\ket{\tilde b}=\bra{ a}\ket{\tilde b}$, and hence, $\bra{ a_k}V_1\ket{\tilde b}=\bra{ a_k}\ket{\tilde b}=b_k$, for $k=0,\dots,d-1$. Therefore,
\begin{equation}
V_1\ket{\tilde b}=\sum_{k=0}^{d-1}b_k\ket{a_k}+b_d\ket{b'_\perp},
\end{equation}
where $\ket{b'_\perp}\in\mathcal A_\perp$ and $\bra{ b'_\perp}\ket{b'_\perp}=\bra{ b_\perp}\ket{b_\perp}$. It follows that there exists a unitary $V_\perp$ acting on $\mathcal A_\perp$ such that $V_\perp\ket{b_\perp}=\ket{b'_\perp}$. Now we set $V=\Pi_{\mathcal A}\oplus V_\perp$, where $\Pi_{\mathcal A}$ is the projection to subspace $\mathcal A$.
\end{proof}
\end{lemma}

\begin{lemma}\label{lemma:set}
Let $\ket{\phi}$ be a state in $\mathcal H_C$ and $\{\ket{\psi_j}\}_{j=0}^{k-1}$ be a set of states in $\mathcal H_S$. If a covariant unitary $U$ acting on $\mathcal H_S\otimes\mathcal H_C$ induces the state conversion
\begin{equation}
\label{eq:equi1} U(\ket{\psi_j}\otimes\ket{\phi})=\ket{\psi'_j}\otimes\ket{\phi},\ j=0,\dots,k-1,
\end{equation}
then there is a covariant unitary $V$ acting on $\mathcal H_S$ such that $V\ket{\psi_j}=\ket{\psi'_j},\ \forall j=0,\dots,k-1$.
\begin{proof}
For $k=1$, it is obvious from Lemma \ref{lemma:equivalent}.

Now we assume that it holds for $k=n$ ($n\leq\dim(H_S)-1$) and prove that it holds for $k=n+1$. From the assumption, a covariant $V_a$ exists such that $\ket{\psi'_j}=V_a\ket{\psi_j}$ for $0\leq j\leq n-1$. The condition as in Eq. (\ref{eq:equi1}) is then written as
\begin{eqnarray}
\label{eq:equi2} U_a(\ket{\psi_j}\otimes\ket{\phi})&=&\ket{\psi_j}\otimes\ket{\phi},\ j=0,\dots,n-1,\\
\label{eq:equi3} U_a(\ket{\psi_n}\otimes\ket{\phi})&=&\ket{\psi''_n}\otimes\ket{\phi},
\end{eqnarray}
where $U_a=(V_a^\dagger\otimes\iden_C)U$ is still a global covariant unitary, and $\ket{\psi''_n}=V_a^\dagger\ket{\psi'_n}$. Now we define a state $\ket{\psi}=c_a\ket{\psi_j}+c_b\ket{\psi_n}$, where $c_a,c_b\neq0$ and $\ket{\psi_j}$ is chosen arbitrarily from $\{\ket{\psi_j}\}_{j=0}^{n-1}$. From Eqs. (\ref{eq:equi2}) and (\ref{eq:equi3}), we have
\begin{equation}
\label{eq:equi4} U_a(\ket{\psi}\otimes\ket{\phi})=\ket{\psi'}\otimes\ket{\phi},
\end{equation}
where $\ket{\psi'}=c_a\ket{\psi_j}+c_b\ket{\psi''_n}$. From Lemma \ref{lemma:equivalent}, Eqs. (\ref{eq:equi3}) and (\ref{eq:equi4}) imply that $\ket{\psi''_n}=V_b\ket{\psi_n}$ and $\ket{\psi'}=V_0\ket{\psi}$, respectively, where $V_b$ and $V_0$ are covariant unitary operations. Hence,
\begin{equation}
\label{eq:equi5} V_0(c_a\ket{\psi_j}+c_b\ket{\psi_n})=c_a\ket{\psi_j}+c_bV_b\ket{\psi_n}.
\end{equation}

Let $\mathcal H^{(i)}$ be the $i$th energy eighenspace of the system, and then $\mathcal H_S=\oplus_i\mathcal H^{(i)}$. One can write $\ket{\psi_j}=\sum_i\ket{\tilde a_j^{(i)}},\ \ket{\psi_n}=\sum_i\ket{\tilde b^{(i)}}$
where $\ket{\tilde a_j^{(i)}},\ket{\tilde b^{(i)}}\in\mathcal H^{(i)}$ are not necessarily normalized. Because $[V_b,H_S]=[V_0,H_S]=0$, we have $V_b=\oplus_iV_b^{(i)}$ and $V_0=\oplus_iV_0^{(i)}$, where $V_b^{(i)}$ and $V_0^{(i)}$ are unitary operators acting on $\mathcal H^{(i)}$. Eq. (\ref{eq:equi5}) is then rewritten as
\begin{equation}
\label{eq:equi6} V_0^{(i)}\left(c_a\ket{\tilde a_j^{(i)}}+c_b\ket{\tilde b^{(i)}}\right)=c_a\ket{\tilde a_j^{(i)}}+c_bV_b^{(i)}\ket{\tilde b^{(i)}},\ \forall i.
\end{equation}
This means that the states $c_a\ket{\tilde a_j^{(i)}}+c_b\ket{\tilde b^{(i)}}$ and $c_a\ket{\tilde a_j^{(i)}}+c_bV_b^{(i)}\ket{\tilde b^{(i)}}$ are unitarily equivalent, so their norms are equal. By noting that this argument holds for all coefficients $c_a$ and $c_b$, and for arbitrary choice of $\ket{\psi_j}$, we arrive at
\begin{equation}
\bra{\tilde a_j^{(i)}}\ket{\tilde b^{(i)}}=\bra{\tilde a_j^{(i)}}V_b^{(i)}\ket{\tilde b^{(i)}},\ \forall i,j.
\end{equation}
Then by Lemma \ref{lemma:pre}, a covariant unitary $V_1$ exists such that $V_1\ket{\psi_j}=\ket{\psi_j}=V_a^\dagger\ket{\psi'_j}$ for $j=0,\dots,n-1$, and $V_1\ket{\psi_n}=\ket{\psi''_n}=V_a^\dagger\ket{\psi'_n}$. By setting $V=V_aV_1$, we find that $\ket{\psi'_j}=V\ket{\psi_j}$, for $j=0,\dots,n$, i.e., this lemma holds for $k=n+1$. This completes the proof.
\end{proof}
\end{lemma}

\begin{lemma}\label{lemma:mix}
For any two states $\rho$ and $\rho'$, and a given pure catalytic state $\ket{\phi}$ of finite dimension, if there exists a covariant unitary $U$, which satisfies $[U,H_S+H_C]=0$, such that
\begin{equation}
\label{eq:equi_mix}U(\rho\otimes\ket{\phi}\bra{\phi})U^\dagger=\rho'\otimes\ket{\phi}\bra{\phi},
\end{equation}
then there exists a covariant unitary $V$ satisfying $[V,H_S]=0$, such that $V\rho V^\dagger=\rho'$.
\begin{proof}
Given $\rho=\sum_jp_j\ket{\psi_j}\bra{\psi_j}$, Eq. (\ref{eq:equi_mix}) is equivalent to
\begin{equation}
\sum_jp_j\ket{\Psi_j}\bra{\Psi_j}=\rho'\otimes\ket{\phi}\bra{\phi},
\end{equation}
where $\ket{\Psi_j}=U(\ket{\psi_j}\otimes\ket{\phi})$. By taking a partial trace on $S$, we have $\sum_jp_j\Tr_S(\ket{\Psi_j}\bra{\Psi_j})=\ket{\phi}\bra{\phi}$, and hence, $\Tr_S(\ket{\Psi_j}\bra{\Psi_j})=\ket{\phi}\bra{\phi}$ for each $j$. This means that each $\ket{\Psi_j}$ is a product state, i.e., $U(\ket{\psi_j}\otimes\ket{\phi})=\ket{\psi'_j}\otimes\ket{\phi},\ \forall j$. By Lemma \ref{lemma:set}, a covariant unitary $V$ exists such that $\ket{\psi'_j}=V\ket{\psi_j},\ \forall j$. It is obvious that $\rho'=\sum_j p_j\ket{\psi'_j}\bra{\psi'_j}$, so we have $\rho'=V\rho V^\dagger$.
\end{proof}
\end{lemma}

Now we are ready to present the proof of Theorem \ref{th1}, which we repeat as follows.
\begin{theorem}\label{appth1}
If $\rho$ cannot be transformed to $\rho'$ under TIO, then the transformation $\rho\otimes\ket{\phi}\bra{\phi}\mapsto\rho'\otimes\ket{\phi}\bra{\phi}$ under TIO is also not achievable for any choice of pure state $\ket{\phi}$.
\begin{proof}
From Proposition 2 in Ref. \cite{PhysRevA.94.052324}, every TIO can be implemented by coupling the system to an ancilla $E$ prepared in a symmetric state via a covariant unitary. If the transformation $\rho\otimes\ket{\phi}\bra{\phi}\mapsto\rho'\otimes\ket{\phi}\bra{\phi}$ under TIO is achievable, then we have
\begin{equation}
\rho'\otimes \ket{\phi}\bra{\phi}=\Tr_E[U(\rho\otimes\ket{\phi}\bra{\phi}\otimes\tau_E)U^\dagger],
\end{equation}
where $[\tau_E,H_E]=0$ and $[U,H_S+H_C+H_E]=0$. By taking a partial trace on $S$, we have $\Tr_{SE}[U(\rho\otimes\ket{\phi}\bra{\phi}\otimes\tau_E)U^\dagger]=\ket{\phi}\bra{\phi}$, and hence $U(\rho\otimes\ket{\phi}\bra{\phi}\otimes\tau_E)U^\dagger=\rho'_{SE}\otimes\ket{\phi}\bra{\phi}$. 
By Lemma \ref{lemma:mix}, a covariant unitary $V$ acting on $SE$ exists such that $V(\rho\otimes\tau_E)V^\dagger=\rho'_{SE}$. By taking a partial trace on $E$, we have $\rho'=\Tr_E(\rho'_{SE})=\Tr_E\left[V(\rho\otimes\tau_E)V^\dagger\right]$, which means that $\rho'$ can be prepared from $\rho$ via TIO.
\end{proof}
\end{theorem}

\section{ASYMMETRY MONOTONES OF QUBIT STATES}\label{app:measure}
Here we briefly review several measures of asymmetry (which are monotonic under TIOs), and then compare their ordering for qubit states.

Consider a system with Hamiltonian $H$ and in state $\rho$. Generally, the measure of asymmetry is a function of both $\rho$ and $H$. In the regime with which we are concerned here, the Hamiltonian $H$ is fixed. Hence in the following, we express the measures of asymmetry as functions of $\rho$.

The quantum Fisher information \cite{clock2003} is defined as
\begin{equation}
I_F(\rho):=\Tr(\dot{\rho}\Delta_\rho^{-1}\dot{\rho}),
\end{equation}
where $\Delta_\rho B:=(\rho B+B\rho)/2$ and $\dot\rho=i[\rho,H]$. 
It quantifies the accuracy of a quantum clock specified by the pair $(\rho,H)$.
A related measure is the Wigner-Yanase skew information \cite{Marvian12}, which is defined as
\begin{equation}
    I_{WY}(\rho):=-\frac12\Tr\left([\rho^{\frac12},H]^2\right).
\end{equation}
It has been proved that both $I_F$ and $I_{WY}$ are additive on tensor products. Further, these two measures become equivalent for pure states, i.e., $I_{WY}(\ket{\psi}\bra{\psi})=\frac14I_F(\ket{\psi}\bra{\psi})=\bra{\psi}H^2\ket{\psi}-(\bra{\psi}H\ket{\psi})^2$.

Note that in the resource theory of asymmetry, the set of free states is convex. Hence, one can employ some general resource measures to quantify the amount of asymmetry, e.g., the robustness and the distance-based measure. The robustness of asymmetry \cite{PhysRevA.93.042107} is defined as
\begin{equation}
R(\rho):=\inf_{\gamma\in\mathcal D}\left\{s:\frac{\rho+s\gamma}{1+s}\in\mathcal F\right\},
\end{equation}
where $\mathcal D$ is the set of all states of the system and $\mathcal F=\{\rho:\rho=e^{-iHt}\rho e^{iHt}\}$ is the set of symmetric states. It is quantitatively related to the task of state discrimination. The distance-based measure is defined as the minimum distance from state $\rho$ to the set of symmetric states $\mathcal{F}$. Formally, let $D(\cdot||\cdot)$ be a distance measure of states, and the distance-based measure of asymmetry is defined as $A_D(\rho):=\min_{\xi\in\mathcal{F}}D(\rho||\xi)$. When the distance measure $D(\rho||\xi)$ is chosen to be the relative entropy $S(\rho||\xi):=\Tr(\rho\log_2\rho-\rho\log_2\xi)$, the corresponding asymmetry measure is called the relative entropy of asymmetry. It has been shown that the relative entropy of asymmetry can be expressed as \cite{Marvian2014}
\begin{equation}
    A_r(\rho)=S[\Pi(\rho)]-S(\rho),
\end{equation}
where $S(\rho):=-\Tr(\rho\log_2\rho)$ is the von Neumann entropy, and $\Pi(\rho)\equiv\sum_j\Pi_j\rho\Pi_j$ with $\Pi_j$ the projection to the $j$-th eigenspace of $H$.

When the system under consideration is a qubit with fixed Hamiltonian $H=\frac{\Delta}{2}\sigma^z$, its asymmetry is relevant to the quantum coherence between the two eigenstates of $H$. In general, a qubit state can be expressed in the Bloch representation as $\rho=\frac{\iden}{2}+\frac{r}{2}\hat{r}\cdot\vec{\sigma}$, where $r\in[0,1]$, $\vec{\sigma}=(\sigma^x,\sigma^y,\sigma^z)$, and $\hat{r}=(\sin\theta\cos\phi,\sin\theta\sin\phi,\cos\theta)$ is a normalized three-dimensional real vector. Direct calculations lead to the results
\begin{eqnarray}
R(\rho)&=&r\sin\theta,\\
I_F(\rho)&=&\Delta^2(r\sin\theta)^2=\Delta^2[R(\rho)]^2,\\
I_{WY}(\rho)&=&\frac{\Delta^2}{4}(1-\sqrt{1-r^2})\sin^2\theta\nonumber\\
&=&\frac{\Delta^2}{4}\frac{[R(\rho)]^2}{1+\sqrt{1-r^2}},\\
A_r(\rho)&=&h(r\cos\theta)-h(r),
\end{eqnarray}
where the function $h(x)\equiv-\frac{1+x}{2}\log_2\frac{1+x}{2}-\frac{1-x}{2}\log_2\frac{1-x}{2}$. 
From these results, we have the following observations.

\emph{Observation 1:} Orderings of states.
Let $\rho_1$ and $\rho_2$ be two qubit states. Then $I_F(\rho_1)\geq I_F(\rho_2)$ is equivalent to $R(\rho_1)\geq R(\rho_2)$. However, it is possible that $I_{WY}(\rho_1)<I_{WY}(\rho_2)$ and/or $A_r(\rho_1)<A_r(\rho_2)$.
This means that, for qubit states with fixed Hamiltonians, the quantum Fisher information $I_F$ and the robustness $R$ give the same ordering of states, while the Wigner-Yanase skew information $I_{WY}$ and the relative entropy of asymmetry $A_r$ give other orderings of states. Therefore, the monotonicity of any measure is a necessary but not sufficient condition for state transformations under TIOs.

\emph{Observation 2:} The maximally asymmetric states.
All of the measures discussed above reach maximum for the set of states $\{\rho|\rho=\frac{1}{2}(\iden+\cos\phi\sigma^x+\sin\phi\sigma^y),\phi\in[0,2\pi)\}$, which we call the maximally asymmetric states. Note that each maximally asymmetric state can be obtained from the state $\rho^+\equiv\frac12(\iden+\sigma^x)$ by a covariant unitary operation.
Also note that the TIO cone of $\rho^+$ does \emph{not} include all the qubit states.

\section{MODES OF COHERENCE AND A GENERAL FORM OF TIOS}\label{app:mode}

Consider a system with Hamiltonian $H=\sum_jE_j\ket{j}\bra{ j}$.
For a quantum state $\rho$ expanded in its energy eigenbasis $\rho=\sum_{i,j}\rho_{ij}\ket{i}\bra{ j}$, a mode of coherence \cite{PhysRevX.5.021001} is defined as \cite{PhysRevA.90.062110}
\begin{equation}
    \rho^{(\delta)}:=\sum_{i,j:E_i-E_j=\delta}\rho_{ij}\ket{i}\bra{ j}.
\end{equation}
Here we define matrices
\begin{equation}\label{eq:pdelta}
    P^{(\delta)}:=\sum_{i,j:E_i-E_j=\delta}\ket{i}\bra{ j},
\end{equation}
and then the modes of coherence can be written as
\begin{equation}
    \rho^{(\delta)}=\rho\odot P^{(\delta)},
\end{equation}
where the label $\odot$ denotes the Hadamard product, i.e., the entrywise matrix product.

Let $\mathcal U_t(\cdot):=e^{-iHt}\cdot e^{iHt}$ denote the free evolution of the system under its Hamiltonian $H$, and we can directly check that
\begin{equation}
    \mathcal{U}_t(\rho^{(\delta)})=e^{-i\delta t}\rho^{(\delta)},
\end{equation}
and then,
\begin{eqnarray}
\mathcal U_t(\rho)&=&\sum_\delta\mathcal U_t(\rho^{(\delta)})\nonumber\\
&=&\sum_\delta e^{-i\delta t}\rho^{(\delta)}\nonumber\\
&=&\sum_\delta e^{-i\delta t}P^{(\delta)}\odot\rho\nonumber\\
&=& T_t\odot\rho,
\end{eqnarray}
where $T_t\equiv\sum_\delta e^{-i\delta t}P^{(\delta)}$.

A TIO operation $\cptp$ satisfies $\cptp\circ\mathcal{U}_t=\mathcal{U}_t\circ\cptp$, which is equivalent to
\begin{equation}
    \sum_\delta e^{-i\delta t}\cptp(P^{(\delta)}\odot\rho)=\sum_\delta e^{-i\delta t}P^{(\delta)}\odot\cptp(\rho),\forall\rho,t.
\end{equation}
Then we have
\begin{equation}\label{eq:commu}
    \cptp\left(P^{(\delta)}\odot\rho\right)=P^{(\delta)}\odot\cptp(\rho),\ \forall\rho.
\end{equation}
This means that, by a TIO operation, each mode in the initial state is independently mapped to the corresponding mode of the final state.

The Choi–Jamiołkowski matrix of operation $\mathcal{E}$ is defined as
\begin{equation}
    J_\cptp=\left(\begin{array}{ccccc}
    \cptp(\ket{0}\bra{0}) & \cdots & \cptp(\ket{0}\bra{j}) & \cdots &\cptp(\ket{0}\bra{d})\\
    \vdots & \ddots & \vdots & \ddots & \vdots \\
    \cptp(\ket{i}\bra{0}) &\cdots & \cptp(\ket{i}\bra{j})  & \cdots &\cptp(\ket{i}\bra{d}) \\
    \vdots & \ddots & \vdots & \ddots & \vdots \\
    \cptp(\ket{d}\bra{0}) &\cdots & \cptp(\ket{d}\bra{j})  & \cdots &\cptp(\ket{d}\bra{d})
    \end{array}
    \right).
\end{equation}
When $\cptp$ is a TIO, then we have
\begin{equation}
    \cptp(\ket{i}\bra{ j})=\cptp(P^{(\delta_{ij})}\odot\ket{i}\bra{ j})=P^{(\delta_{ij})}\odot\cptp\left(\ket{i}\bra{j}\right),
\end{equation}
where $\delta_{ij}=E_i-E_j$. Here the first equation is from the definition of $P^{(\delta)}$, and the second equation is from Eq. (\ref{eq:commu}). Now we define a matrix
\begin{equation}\label{eq:p_matrix}
    P:
    =\left(\begin{array}{ccccc}
    P^{(\delta_{00})} & \cdots & P^{(\delta_{0j})} & \cdots & P^{(\delta_{0d})}\\
    \vdots & \ddots & \vdots & \ddots & \vdots \\
    P^{(\delta_{i0})} &\cdots & P^{(\delta_{ij})}  & \cdots &P^{(\delta_{id})} \\
    \vdots & \ddots & \vdots & \ddots & \vdots \\
    P^{(\delta_{d0})} &\cdots & P^{(\delta_{dj})}  & \cdots &P^{(\delta_{dd})}
    \end{array}
    \right).
\end{equation}
Then the Choi–Jamiołkowski matrix of a TIO operation $\cptp$ satisfies
\begin{equation}\label{eq:cj_tio}
    J_\cptp=J_\cptp\odot P.
\end{equation}
This is the general form of a TIO operation.

As an example, we consider a qubit system with Hamiltonian $H=\frac{\Delta}{2}\sigma^z$. The matrix $P$ defined in Eq. (\ref{eq:p_matrix}) reads
\begin{equation}
    P:=\left(\begin{array}{cc}
    P^{(0)} & P^{(\Delta)}  \\
    P^{(-\Delta)} & P^{(0)} 
    \end{array}
    \right)
\end{equation}
with $P^{(0)}=\iden$ and $P^{(\pm\Delta)}=\frac12(\sigma^x\pm i\sigma^y)$. Then from Eq. (\ref{eq:cj_tio}), the Choi–Jamiołkowski matrix of a qubit TIO is generally written as
\begin{equation}\label{eq:cj_qubit}
    J_\cptp=\left(\begin{array}{cccc}
    p_0 & 0 & 0 & \gamma \\
    0 & 1-p_0 & 0 & 0 \\
    0 & 0 & 1-p_1 & 0 \\
    \gamma^* & 0 & 0 & p_1
    \end{array}
    \right),
\end{equation}
where the parameters satisfy $p_0,p_1\in[0,1]$ and $|\gamma|\leq\sqrt{p_0 p_1}$, such that $J_\cptp$ is positive.

\section{TIO CONE AND CCTIO CONE OF A QUBIT STATE}\label{app:cone}

In the Bloch presentation, a qubit state is generally written as $\rho(\vec r)=\frac12(\iden+\vec{r}\cdot\vec{\sigma})$, where $\vec{r}$ is a three-dimensional real vector with $\left|\vec{r}\right|\leq1$ and $\vec{\sigma}=(\sigma^x,\sigma^y,\sigma^z)$. 
The basic structure of its TIO cone or CCTIO cone is that it is rotationally symmetric about the $z$ axis. 
This is because any set of states which are rotationally symmetric about the $z$ axis are equivalent by covariant unitary operators $U(\phi)=\diag(1,e^{i\phi})$.

Let the Bloch vector $\vec{r}=(\eta\cos\phi,\eta\sin\phi,z)$ with $\eta\in[0,1],\phi\in[0,2\pi),z\in[-1,1]$. The TIO cone of $\rho(\vec{r})$ is written as
\begin{widetext}
\begin{eqnarray}\label{eq:tiocone_qubit}
\tiocone[\rho(\vec{r})]&=&\bigg\{\rho':\rho'=\frac12(\iden+\eta'\cos\phi'\sigma^x+\eta'\sin\phi'\sigma^y+z'\sigma^z),\nonumber\\
&&\quad\quad 0\leq\eta'\leq\min\left\{\eta\sqrt{\frac{1+z'}{1+z}},\eta\sqrt{\frac{1-z'}{1-z}}\right\}, z'\in[-1,1], \phi'\in[0,2\pi)\bigg\}.
\end{eqnarray}
\end{widetext}
The reason is as follows. After the action of a TIO in the form of Eq. (\ref{eq:cj_qubit}), the qubit state $\rho(\vec{r})$ becomes
\begin{equation}
    \rho'=\frac12\left(\begin{array}{cc}
    1+z' & \eta'e^{-i\phi'}  \\
    \eta'e^{i\phi'} & 1-z' 
    \end{array}
    \right)
\end{equation}
with
\begin{eqnarray}
\label{eq:z} z'  =  p_0(1+z)-p_1&(1-z)-z,\\
\label{eq:eta} \eta'e^{i\phi'} = \gamma\eta e^{i\phi}&.
\end{eqnarray}
From Eq. (\ref{eq:z}) and $p_0,p_1\in[0,1]$, we have
\begin{eqnarray}
    p_0&\in&[0,1]\cap\left[\frac{z+z'}{1+z},\frac{1+z'}{1+z}\right],\nonumber\\
\label{eq:p0p1}    p_1&=&\frac{(1+z)p_0-(z+z')}{1-z}.
\end{eqnarray}
From Eq. (\ref{eq:eta}) and $|\gamma|\leq\sqrt{p_0 p_1}$, we have
\begin{eqnarray}
\eta'&=&|\gamma|\eta\leq\eta\sqrt{p_0p_1}\nonumber\\
& \leq & \eta\cdot\min\left\{\sqrt{\frac{1+z'}{1+z}},\sqrt{\frac{1-z'}{1-z}}\right\}.
\end{eqnarray}
Here the last inequality is from Eq. (\ref{eq:p0p1}). 
The extreme states, for which the above equality holds, are obtained with $|\gamma|=\sqrt{p_0p_1}$, $p_1$ as in Eq. (\ref{eq:p0p1}), and
\begin{equation}
    p_0=\min\left\{1,\frac{1+z^\prime}{1+z}\right\}
\end{equation}
From the rotational symmetry and the convexity of TIO cone $\tiocone[\rho(\vec{r})]$, we arrive at Eq. (\ref{eq:tiocone_qubit}). 
This completes the proof.

Next, we prove the following statement.
For state $\rho(\eta)=\frac12(\iden+\eta\sigma^x)$, the CCTIO cone $\cccone^{(d)}$ is symmetric about the $xy$ plane.
The reason is as follows. 
Suppose the $\rho^\uparrow=\frac12(\iden+r_x\sigma_x+r_z\sigma_z)\in\cccone^{(d)}[\rho(\eta)]$, namely, a covariant operation $\cptp^\uparrow$ and a $d$-dimensional catalyst in state $\sigma^\uparrow$ exist such that $\cptp^\uparrow[\rho(\eta)\otimes\sigma^\uparrow]=\rho^\uparrow|\sigma^\uparrow$. 
Let $\sigma^\downarrow=U^x\sigma^\uparrow U^x$ and $\cptp^\downarrow=\mathcal U_x\circ\cptp^\uparrow\circ\mathcal U_x$ with $\mathcal U^x(\cdot)\equiv \sigma^x\otimes U^x(\cdot)\sigma^x\otimes U^x$. Here the unitary operator $U^x$ reverses the energy levels of the catalyst $C$, i.e., it is anti-diagonal on the eigenbasis of $H_C$ with each non-zero entry equal to 1. 
It is directly checked that $\cptp^\downarrow\in\tio$ and $\cptp^\downarrow[\rho(\eta)\otimes\sigma^\downarrow]=\rho^\downarrow|\sigma^\downarrow$ with $\rho^\downarrow=\frac12(\iden+r_x\sigma_x-r_z\sigma_z)$, i.e., $\rho^\downarrow\in\cccone^{(d)}[\rho(\eta)]$.

\section{NUMERICAL METHOD ON EVALUATING CCTIO CONE OF QUBIT STATES}\label{app:numerical}
Here we set the input states and target states as $\rho_S$ and $\rho^\prime_S$, respectively.
In the Bloch representation, the extreme states in the CCTIO cone are defined as those with the maximum distance from the $z$ axis for a given $r'_z\equiv\Tr(\sigma^z\rho'_S)$.
Since any state in the cone can be achieved by applying a dephasing map (which is a TIO) on an extreme state, it is sufficient to solve the extreme states to obtain the whole cone.
Because the CCTIO cone of a qubit system is symmetric about the $z$ axis, we only need to solve the extreme states within the $xz$ plane with $x\geq0$.
Our problem then becomes the optimization task
\begin{equation}\label{eq:fixedupper}
        R_{\mathrm{CC}}^{(d)}(\rho_S,r^\prime_z) = \max\limits_{\sigma_C\in D(\mathcal{H}^{(d)}_C)}R_{\mathrm{CC}}(\sigma_C;\rho_S,r^\prime_z),
    \end{equation}
where $D(\mathcal{H}^{(d)}_C)$ is the set of $d$-dimensional density matrices, and
the function $R_{\mathrm{CC}}(\sigma_C;\rho_S,r^\prime_z)$ is defined as 
    \begin{equation}\label{eq:fixedlower}
        \begin{aligned}
        R_{\mathrm{CC}}(\sigma_C&;\rho_S,r^\prime_z) = \max_{\mathcal{E}\in\tio}\Tr[\sigma^x\rho'_S]\\
        \mathrm{s.t.}\quad &\rho^\prime_S = \Tr_C\left[\mathcal{E}(\rho_S\otimes\sigma_C)\right],\\
        &\sigma_C \equiv \Tr_S\left[\mathcal{E}(\rho_S\otimes\sigma_C)\right],\\
            &r^\prime_z \equiv  \Tr\left[\sigma^z\rho^\prime_S\right].
        \end{aligned}
    \end{equation}

Clearly, the function $R_{\mathrm{CC}}(\rho_S,r^\prime_z)$ embeds a lower-level optimization, Eq. (\ref{eq:fixedlower}), into an upper-level optimization,  Eq. (\ref{eq:fixedupper}). 
This optimization task is called bi-level optimization \cite{bilevelopt}, and is generally hard to be solved. 
Here we first consider the lower level of optimization, and then describe the methods for solving Eq. (\ref{eq:fixedupper}) for $d=2,3$.

The lower-level optimization as in Eq. (\ref{eq:fixedlower}) can be reformulated as a semidefinite programming (SDP) task, which allows us to effectively solve it in polynomial time via interior point methods \cite{boyd2004convex}.
Here, we derive the explicit SDP form in the following.
Let $H_S=\frac{\Delta}{2}\sigma^z$ be the Hamiltonian of $S$ and $H_C=\sum_{l=0}^{d-1}l\Delta\ket{l}\bra{ l}$ be the Hamiltonian of $C$. The eigenbasis of the total Hamiltonian $H_{SC}=H_S+H_C$ is labeled $\left\{\ket{\psi_j}\right\}^{2d-1}_{j=0}$, and the eigenvalue of each $\ket{\psi_j}$ is denoted $E_j$.
    Thus, the explicit form of matrix $P$ in Eq. (\ref{eq:cj_tio}) is written as
    \begin{equation}
        P = \sum_{jk}\ket{\psi_j}\bra{\psi_k}\otimes P^{(m_{jk}\Delta)},
    \end{equation}
    where $P^{(m_{jk}\Delta)}$ is defined in Eq. (\ref{eq:pdelta}) with $m_{jk}\Delta \equiv E_j-E_k$.
    Note that $m_{jk}$ are integers and satisfy $-d\leq m_{jk}\leq d$.
    Then, by setting the optimization variable as the Choi–Jamiołkowski matrix $J_{\mathcal{E}}$ of TIO $\mathcal{E}$, we arrive at the SDP form of Eq. (\ref{eq:fixedlower}) as
    \begin{equation}\label{eq:semide}
        \begin{aligned}
        \max\limits_{J_\mathcal{E}}\quad&\Tr\left[\sigma^x\rho^\prime_S\right],\\
        \mathrm{s.t.}\quad & J_\mathcal{E}\geq 0,\quad\textnormal{(CP condition)},\\
        & \Tr_{S^\prime C^\prime}\left[J_{\mathcal{E}}\right] = \iden_{SC},\quad\textnormal{(TP condition)},\\
        & J_{\mathcal{E}}\odot P = J_{\mathcal{E}},\quad\textnormal{(TIO condition)},\\
        & \rho^\prime_{SC}=\Tr_{SC}\left[\left(\rho_S\otimes\sigma_C\otimes\iden_{S^\prime C^\prime}\right)^\mathrm{T}\cdot J_{\mathcal{E}}\right],\\
        & \sigma_C \equiv \Tr_{S^\prime}\left[\rho^\prime_{SC}\right], \quad\textnormal{(CC condition)} ,\\
        & \rho^\prime_S = \Tr_{C^\prime}\left[\rho^\prime_{SC}\right],\quad r^\prime_z\equiv\Tr(\sigma^z\rho^\prime_S),
        \end{aligned}
    \end{equation}
    where the CC condition denotes the correlating-catalyst condition, $S^\prime C^\prime$ is the output space of $SC$, and the total target states $\rho^\prime_{SC}$ follow the definition of the Choi–Jamiołkowski matrix.
    In practice, we use the CVX package \cite{cvx} to numerically solve this SDP task with tolerance at $1.81\times 10^{-12}$.

For the upper part of optimization, the property of $R_{\mathrm{CC}}(\sigma_C;\rho_S,r^\prime_z)$ is essential. 
In the following, we prove the continuity of $R_{\mathrm{CC}}(\sigma_C;\rho_S,r^\prime_z)$ on $\sigma_C$, 
which allows us to find the optimizer of $R_{\mathrm{CC}}(\rho_S,r^\prime_z)$ over $D(\mathcal{H}_C)$ by sampling.

\begin{lemma}[$ R_{\mathrm{CC}}(\sigma_C;\rho_S,r^\prime_z)$ has Lipschitz continuity on $\sigma_C$.]\label{lm:ac}
    For any pair of catalytic states $\sigma_C$ and $\sigma^{\epsilon}_C$ satisfying $\|\sigma_C-\sigma^{\epsilon}_C\|_1\leq\epsilon$, we have
    \begin{equation}\label{eq:repsilon}
        \left| R_{\mathrm{CC}}(\sigma^\epsilon_C;\rho_S,r^\prime_z)- R_{\mathrm{CC}}(\sigma_C;\rho_S,r^\prime_z)\right|\leq 4\epsilon(1+\epsilon).
    \end{equation}
    Here $\|\cdot\|_1$ denotes the trace-norm of states.
\end{lemma}
\begin{proof}
    For convenience, we set 
    \begin{equation}
         R_{\mathrm{CC}}(\sigma^\epsilon_C;\rho_S,r^\prime_z)\geq R_{\mathrm{CC}}(\sigma_C;\rho_S,r^\prime_z).
    \end{equation}
    One optimizer of $ R_{\mathrm{CC}}(\sigma^\epsilon_C;\rho_S,r^\prime_z)$ is denoted by $\mathcal{E}^\star$,
    and the corresponding output state of $S$ is $\rho^{\star}_S$.

The proof is sketched as follows. First, we construct a trace preserving (TP) map $\mathcal{N}$, which is close to $\mathcal{E}^\star$, and satisfies
\begin{equation}\label{eq:N_map}
    \mathcal{N}(\rho_S\otimes\sigma_C) = \rho^\star_S|\sigma_C.
\end{equation}
Second, we slightly extend the set of free operations which lead to an upper bound of $R_{\mathrm{CC}}(\sigma_C^\epsilon;\rho_S,r^\prime_z)$.
Finally, we prove that the difference between $R_{\mathrm{CC}}(\sigma_C;\rho_S,r^\prime_z)$ and the upper bound of $R_{\mathrm{CC}}(\sigma_C^\epsilon;\rho_S,r^\prime_z)$ is no larger than $4\epsilon(1+\epsilon)$.
    
Let us construct a TP map (which need not be completely positive) as
    \begin{equation}\label{eq:N_def}
            \mathcal{N}=\mathcal{I}\otimes\mathcal{N}_1\circ\mathcal{E}^\star\circ\mathcal{I}\otimes\mathcal{N}_0.
    \end{equation}
    Here $\mathcal{I}$ denotes the identity map, $\mathcal{N}_0 := \mathcal{I}+\mathcal{M}_0$, and $\mathcal{N}_1 := \mathcal{I} + \mathcal{M}_1$,
    where $\mathcal{M}_0$ and $\mathcal{M}_1$ are constant maps defined as
    \begin{equation}
    \mathcal{M}_0(\cdot) = \sigma_C^\epsilon-\sigma_C,\quad\mathcal{M}_1(\cdot) = \sigma_C-\sigma_C^\epsilon.
    \end{equation}
     It is easy to check that $\mathcal{N}$ satisfies Eq. (\ref{eq:N_map}), and that
    \begin{equation}
        \|\mathcal{I}-\mathcal{N}_0\|_{\diamond}\leq \epsilon,\quad \|\mathcal{I}-\mathcal{N}_1\|_{\diamond}\leq \epsilon,
    \end{equation} 
    where $\|\cdot\|_{\diamond}$ is the diamond norm \cite{Aharonov1998}. Then, we examine the distance between $\mathcal{N}$ and $\mathcal{E}^\star$ as
    \begin{eqnarray}\label{eq:NE_diam}
            \|\mathcal{N}-\mathcal{E}^\star\|_{\diamond}&\leq&\|\mathcal{I}\otimes\mathcal{M}_1\circ\mathcal{E}^\star\|_{\diamond} + \|\mathcal{E}^\star\circ\mathcal{I}\otimes\mathcal{M}_0\|_{\diamond}\nonumber\\
            &&+\|\mathcal{I}\otimes\mathcal{M}_1\circ\mathcal{E}^\star\circ\mathcal{I}\otimes\mathcal{M}_0\|_{\diamond}\nonumber\\
            &\leq& 2\epsilon(1+\epsilon).
    \end{eqnarray}
Now we define the set of the allowed bipartite operations for given $\rho_S$, $\sigma_C$, and $r'_z$ as
\begin{eqnarray}
   \mathcal{O}_{(\sigma_C;\rho_S,r'_z)}&=\big\{\cptp:\cptp\in\tio,\Tr_S\left[\cptp(\rho_S\otimes\sigma_C)\right]=\sigma_C,\nonumber\\
    & \Tr\left[\sigma^z\Tr_C\left(\cptp(\rho_S\otimes\sigma_C)\right)\right]=r'_z\big\}.
\end{eqnarray}
Then Eq. (\ref{eq:NE_diam}) implies that the TP map $\mathcal{N}$ is $2\epsilon(1+\epsilon)$-close to the set $\mathcal{O}_{(\sigma_C;\rho_S,r'_z)}$.

Next, we define a set of TP maps as
\begin{equation}\label{eq:tio_epsilon}
\begin{aligned}
        \mathcal{O}^{\epsilon}_{(\sigma_C;\rho_S,r'_z)} 
        = &\big\{\mathcal{E}^{\prime}:\inf\limits_{\mathcal{E}\in\mathcal{O}_{(\sigma_C;\rho_S,r'_z)}}\|\mathcal{E}-\mathcal{E}^\prime\|_{\diamond}\leq \epsilon,\\
        &\ \mathcal{E}^\prime\in\mathrm{TP},\ \Tr_S\left[\cptp'(\rho_S\otimes\sigma_C)\right]=\sigma_C,\\  &\Tr\left[\sigma^z\Tr_C\left(\cptp'(\rho_S\otimes\sigma_C)\right)\right]=r'_z\big\},
\end{aligned}
\end{equation}
and a function as
\begin{equation}
        \begin{aligned}
            R^{\epsilon}_{\mathrm{CC}}(\sigma_C;\rho_S,r^\prime_z)&=\max\limits_{\mathcal{E}^\prime\in\mathcal{O}^{\epsilon}_{(\sigma_C;\rho_S,r'_z)}} \Tr[\sigma^x\rho^{\prime\prime}_S],\\
                \mathrm{s.t.}\quad \rho^{\prime\prime}_S &= \Tr_C\left[\mathcal{E}^\prime(\rho_S\otimes\sigma_C)\right],\\
                \rho^{\prime\prime}_S &\textnormal{ is positive semi-definite,}\\
                \sigma_C &\equiv \Tr_S\left[\mathcal{E}^\prime(\rho_S\otimes\sigma_C)\right],\\
                r^\prime_z&\equiv\Tr(\sigma^z\rho^{\prime\prime}_S).
        \end{aligned}
\end{equation} 
Because $\mathcal{N}\in\mathcal{O}^{2\epsilon(1+\epsilon)}_{(\sigma_C;\rho_S,r'_z)}$, we have
    \begin{equation}\label{eq:2e+es}
        \begin{aligned}
             R_{\mathrm{CC}}(\sigma_C^\epsilon;\rho_S,r^\prime_z) &= \Tr[\sigma^x\rho_S^\star]\\
             &=\Tr\left[\sigma^x\Tr_C\left[\mathcal{N}\left(\rho_S\otimes\sigma_C\right)\right]\right]\\
            &\leq R^{2\epsilon(1+\epsilon)}_{\mathrm{CC}}(\sigma_C;\rho_S,r^\prime_z).
        \end{aligned}
    \end{equation} 
    
Then we turn to the difference between  $R_{\mathrm{CC}}(\sigma_C;\rho_S,r^\prime_z)$ and $R^{2\epsilon(1+\epsilon)}_{\mathrm{CC}}(\sigma_C;\rho_S,r^\prime_z)$.
From Eq. (\ref{eq:tio_epsilon}), for any state $\rho''_S$ which can be obtained as $\rho''_S=\Tr_C[\cptp'(\rho_S\otimes\sigma_C)]$ with $\cptp'\in\mathcal{O}^{2\epsilon(1+\epsilon)}_{(\sigma_C;\rho_S,r'_z)}$, there exists an operation $\cptp\in\mathcal{O}_{(\sigma_C;\rho_S,r'_z)}$, such that $\cptp(\rho_S\otimes\sigma_C)=\rho'_S|\sigma_C$ and
    \begin{equation}\label{eq:ac_distance}
        \|\rho^{\prime}_S-\rho^{\prime\prime}_S\|_1\leq 2\epsilon(1+\epsilon),\ 
    \end{equation}
It follows that
    \begin{equation}\label{eq:ac}
        |\Tr(\sigma^x\rho'_S)-\Tr(\sigma^x\rho''_S)|\leq 2 \|\rho'_S-\rho''_S\|_1\leq4\epsilon(1+\epsilon),
    \end{equation}
and hence,
    \begin{equation}
        R^{2\epsilon(1+\epsilon)}_{\mathrm{CC}}(\sigma_C;\rho_S,r^\prime_z)\leq R_{\mathrm{CC}}(\sigma_C;\rho_S,r^\prime_z)+ 4\epsilon(1+\epsilon).
    \end{equation}
    Recalling Eq. (\ref{eq:2e+es}), we arrive at
    \begin{equation}
         R_{\mathrm{CC}}(\sigma^\epsilon_C;\sigma_S,r^\prime_z) \leq R_{\mathrm{CC}}(\sigma_C;\rho_S,r^\prime_z)+4\epsilon(1+\epsilon).
    \end{equation}
    This completes the proof.
\end{proof}

The continuity of $R_{\mathrm{CC}}(\sigma_C;\rho_S,r^\prime)$ allows us to numerically calculate $R_{\mathrm{CC}}^{(d)}(\rho_S,r^\prime_z)$ [which is the upper bound of $R_{\mathrm{CC}}(\sigma_C^{(d)};\rho_S,r^\prime)$] with a small error, by  uniform sampling of $d$-dimensional catalytic states. 
    For $d=2$, because states which are rotational symmetric about the $z$ axis are equivalent under covariant unitaries, it is sufficient to sample within the $xz$ plane with $x\geq0$. 
    When $d>2$, due to the enormous sampling cost, we can only obtain the local optimal by the gradient descent method, 
    though this continuity property can accelerate the initial sampling process of this optimization task. The technical details of solving $\cccone^{(d)}$ are shown in the following.

    For solving $\cccone^{(2)}$, we uniformly sample on the $xz$ plane with $x\geq 0$, as mentioned before. In practice, we sample on a two-dimensional lattice with constant $1/256$ (i.e., the size of the lattice cell along both the $x$ and the $z$ axis is set to be $1/256$).
    Thus, in a lattice cell, the maximal distance of unsampled points from the sampled point is $2\epsilon =1/(256\sqrt{2})$. Then, according to Eq. (\ref{eq:repsilon}), we obtain the Lipschitz error of the upper bound of $\cccone^{(2)}$ as less than $4\epsilon(1+\epsilon)\approx 5.553\times 10^{-3}$.

    To obtain $\cccone^{(3)}$, we use the gradient descent optimization to find the local maximal of $R_\mathrm{CC}(\sigma^{(3)}_C;\rho_S,r^\prime_z)$. 
    The initial points for the gradient descent are chosen via the LIPO algorithm \cite{Malherbe2017}, which allows us to use the Lipschitz condition to accelerate the searching process.
    Note that we use the Hilbert–Schmidt ensemble generating method \cite{generateRD} as the sub-task of LIPO to randomly sample the catalysts.
    Then we parametrize every three-dimensional catalysts using eight Gell-Mann matrices \cite{PhysRev.125.1067}, such that we can calculate the approximate gradient
    in $\mathbb{R}^8$. The termination tolerance of the function is set at $1\times 10^{-8}$. 
    Note that this strategy can also be applied to obtain $\cccone^{(d)}$ with $d>3$.
    
\bibliography{apssamp}

\end{document}